\documentclass{llncs}
\newcommand{\probleme}[3]{\medskip\noindent\textbf{#1 problem}\\ \noindent \emph{Instance: }#2\\ \noindent\emph{Question: }#3\medskip}
\usepackage{graphicx}
\usepackage{amsmath}
\usepackage{comment}
\usepackage[algoruled,vlined,linesnumbered,english]{algorithm2e}
\newcommand{\cG}{\cal G}
\newcommand{\graph}{{snapshot graph}}
\newcommand{\tree}{{shortest-time tree}}

\newcommand{\DEL}{ELCF}

\newcommand{\fpartb}{N3DM}
\newcommand{\cA}{\mathcal{A}}
\newcommand{\cB}{\mathcal{B}}
\newcommand{\cC}{\mathcal{C}}

\newcommand{\bA}{\mathbb{A}}
\newcommand{\bB}{\mathbb{B}}
\newcommand{\bC}{\mathbb{C}}

\newtheorem{observation}{Observation}
\usepackage{txfonts}

\title{Exploring Graphs with Time Constraints by Unreliable Collections of Mobile Robots}

\author{Jurek Czyzowicz\inst{1}\thanks{Research supported in part by NSERC Discovery grant.}
\and
Maxime Godon\inst{1}
\and
Evangelos Kranakis\inst{2}\thanks{Research supported in part by NSERC Discovery grant.}
\and
Arnaud Labourel\inst{3}\thanks{Research partially supported by the ANR project MACARON (anr-13-js02-0002).}
\and
Euripides Markou\inst{4}
}
\institute{Universit\'e du Qu\'ebec en Outaouais, Gatineau, Canada
\and
School of Computer Science, Carleton University, Ottawa, Canada
\and
LIF, Aix-Marseille University \& CNRS, France
\and
University of Thessaly, Lamia, Greece
}

\begin{document}
\maketitle

\begin{abstract}
A graph environment must be explored by a collection of mobile robots. Some of the robots, a priori unknown, may turn out to be unreliable. The graph is weighted and each node is assigned a deadline. The exploration is successful if each node of the graph is visited before its deadline by a reliable robot. The edge weight corresponds to the time needed by a robot to traverse the edge. Given the number of robots which may crash, is it possible to design an algorithm, which will always guarantee the exploration, independently of the choice of the subset of unreliable robots by the adversary? We find the optimal time, during which the graph may be explored. Our approach permits to find the maximal number of robots, which may turn out to be unreliable, and the graph is still guaranteed to be explored.


We concentrate on line graphs and rings, for which we give positive results. We start with the case of the collections involving only reliable robots. We give algorithms finding optimal times needed for exploration when the robots are assigned to fixed initial positions as well as when such starting positions may be determined by the algorithm. 
We extend our consideration to the case when some number of robots may be unreliable.
Our most surprising result is that solving the line exploration problem with robots at given positions, which may involve crash-faulty ones, is NP-hard. The same problem has polynomial solutions for a ring and for the case when the initial robots' positions on the line are arbitrary.

The exploration problem is shown to be NP-hard for star graphs, even when the team consists of only two reliable robots.
%
%
%
\end{abstract}

\section{Introduction}

Alice and Bob is a busy Ottawa couple with three kids Chris, Donald and Elsa. One day they need to pick up Elsa from the kindergarten, drive Donald to the wrestling practice and get Chris to the train station. They also need to get groceries, pick up wine and flowers before each store closes for a dinner party in their house. How should Alice and Bob share these tasks to minimize the effort and complete each one before its deadline?

An Ottawa School Bus Company needs to transport pupils to local schools before the start of their classes. Given the harsh Canadian climate, it is the norm rather than exception that a number of buses fail to function on any given day and an adequate replacement must be planned in advance. How should the buses allocate the tasks so as to successfully conclude the distribution of students while respecting the time deadlines?   

Throughout this paper, the environment is modelled by a graph that must be serviced by a collection of mobile robots. The graph edges are weighted by numbers, representing the time it takes to traverse them. Each graph node is assigned a deadline, representing the maximal time moment to deliver a service to this node by some mobile robot. A number of robots may crash during their work. What is the minimal time needed to service a given graph by a collection of $k$ robots? What is such a time if we assume that up to $f$ unknown robots may  crash during their work? 

\subsection{Preliminaries and notation}


We are given a weighted $n$-node graph $G = (V, E)$ with $V$ its set of vertices, $E$ its set of edges, and a set of $k$ mobile robots initially placed at a subset of its nodes. The weight of an edge $\{v_i, v_j\}$ corresponds to the time it takes to be traversed by a robot. Each node $v_i$ of the graph is assigned a deadline $\Delta_i$, which is a positive real number. Robots walk along the edges of the graph with unit speed. The robots collaborate attempting to explore the entire graph. However, a subset of up to $f$ robots may turn out to be unreliable and fail to collaborate. Unreliability refers to the robots which may be crash faulty in that they suffer from an (unspecified) passive, omission failure and then stop responding but are otherwise harmless.  This subset of unreliable robots may be chosen by the adversary, which is assumed to know our algorithm beforehand. The exploration is successful if each graph node is visited before its deadline by at least one of the reliable robots. 

We assume that nodes already explored ``do not block passage'' and can still be visited, even after their deadlines have expired, by robots on their way to reaching unexplored parts of the graph.

We denote by $t \to r_i(t)$ the trajectory of the $i$-th robot as a function of the time $t$, where $r_i (t)$ denotes the position of the $i$-th robot in the graph at time $t$,  for $i = 1, 2, \ldots , k$. Note that at a given time $t$, a robot may be located in the interior of an edge.

By a schedule we mean a set of functions $r_i(t), i = 1, 2, \ldots , k$ which define the motion of the robots respecting their maximum unit speed. We say that the schedule \textit{explores} the graph if for each node $v_i$ there exists a robot $r_j$ such that $r_j(t^*) = v_i$, for some time $t^* \leq \Delta_i$ .

Given a time $\Delta$, we study the {\em decision problem} whether the graph may be successfully explored before time $\Delta$. We also look at the {\em optimization problem}, that is, the problem of ensuring that the reliable robots visit every node before expiration of its deadline, and the last explored node is visited as fast as possible. If for any schedule, the adversary can find a subset of $f$ unreliable robots, so that any of the remaining $k-f$ robots fails to visit some node before its deadline, then the instance of the problem is deemed unsolvable.

\subsection{Related work}
Searching a graph with one or more searchers has been widely studied in the mathematics literature (see, e.g. \cite{FT08} for a survey).  There is extensive literature on linear search (referring to searching a line in the continuous or discrete model), e.g., see \cite{baezayates1993searching} for optimal deterministic linear search and \cite{demaine2006online} for algorithms incorporating a \textit{turn cost} when a robot changes direction 
during the search. Variants of search using collections of \textit{collaborating} robots has also been investigated. The robots can employ either \textit{wireless} communication (at any distance) or \textit{face-to-face} communication, where communication is only possible among co-located robots. For example, the problem of \textit{evacuation} \cite{CGKNOV} is essentially a search problem where search is completed only when the target is reached by the last robot. Linear group search in the face-to-face communication model has also been studied with robots that either operate at the same speed or with a pair of robots having distinct maximal speeds \cite{SIROCCO16,Groupsearch}. Linear search with multiple robots where some fraction of the robots may exhibit either {\em crash faults} or \textit{Byzantine faults} is studied in \cite{PODC16} and \cite{ISAAC16}, respectively. 


The (Directed) Rural Postman Problem (DRPP) is a general case of the Chinese Postman Problem where a subset of the set of arcs of a given (directed) graph is 'required' to be traversed at minimum cost. \cite{christofides1986algorithm} presents a branch and bound algorithm for the exact solution of the DRPP based on bounds computed from Lagrangian Relaxation. \cite{corberan1994polyhedral} studies the polyhedron associated with the Rural Postman Problem and characterizes its facial structure. \cite{eiselt1995arc} gives a survey of the directed and undirected rural postman problem and also  discusses applications.

A scheduling problem considered by the research community concerns $n$ jobs, each to be processed by a single machine, subject to arbitrary given precedence constraints; associated with each job $j$ is a known processing time $a_j$ and a monotone nondecreasing cost function $c_j(t)$, giving the cost that is incurred by the completion of that job at time $t$. \cite{lawler1973optimal} gives an efficient computational procedure for the problem of finding a sequence which will minimize the maximum of the incurred costs. Further, \cite{lawler1973optimal} also studies a class of time-constrained vehicle routing and scheduling problems that may be encountered in several transportation/ distribution environments. In the single-vehicle scheduling problem with time window constraints, a vehicle has to visit a set of sites on a graph, and each site must be visited after its ready time but no later than its deadline.  \cite{young1999single} studies the problem of minimizing the total time taken to visit all sites.  \cite{garey1977two} considers the problem of determining whether there exists a schedule on two identical processors that executes each task in the time interval between its start-time and deadline and presents an $O(n^3)$ algorithm that constructs such a schedule whenever one exists.

The author of \cite{bock2015solving} resolves the complexity status of the well-known Traveling Repairman Problem on a line (Line-TRP) with general processing times at the request locations and deadline restrictions by showing that it is strongly NP-complete. \cite{mitrovic2002multiple} considers the problem of finding a lower and an upper bound for the minimum number of vehicles needed to serve all locations of the multiple traveling salesman problem with time windows in two types of precedence graphs: the start-time precedence graph and the end-time precedence graph. \cite{holte1989pinwheel} considers ``the pinwheel'', a formalization of a scheduling problem arising in satellite transmissions whereby a piece of information is transmitted for a set duration, then the satellite proceeds with another piece of information while a ground station receiving from several such satellites and wishing to avoid data loss faces a real-time scheduling problem on whether a ``useful'' representation of the corresponding schedule exists.

The work of \cite{tsitsiklis1992special} is very related to our work in that jobs are located on a line. Each job has an associated processing time, and whose execution has to start within a prespecified time window. The paper considers the problems of minimizing (a) the time by which all jobs are executed (traveling salesman problem), and (b) the sum of the waiting times of the jobs (traveling repairman problem). 
Also related is the research on Graphs with dynamically evolving links (also known as time varying graphs) which has been explored extensively in theoretical computer science (e.g., see~\cite{casteigts2011,flocchini2015time,kuhn2010}).

\subsection{Outline and results of the paper}
We consider first the collections of robots which are all reliable. We start in Section \ref{sct:1onLine}~with the case of a single robot on a line graph and we give an algorithm finding the shortest exploration time when the robot's starting position is given, is arbitrary, or it is arbitrary but restrained to some subset of line nodes. In Section~\ref{sct:ManyInLine} we study line exploration by a collection of robots at fixed or arbitrary positions on the line. We observe, that these algorithms may be extended to the ring case, although their complexity is slightly compromised.

In Section~\ref{sct:LineFaulty} we consider the case of unreliable robots. In one case, we show an unexpected result. If $k$ robots are at given fixed initial positions on the line and up to $f$ out of $k$ robots may turn out to be crash-faulty, the problem of finding the optimal exploration time is NP-hard. This result holds even if the nodes' deadlines may be ignored (e.g. they are infinite for all nodes). For all other settings we give algorithms finding optimal exploration times.
In Section~\ref{ sct:Ring}~we extend our approach to the ring environment. However, the setting which was proven to be NP-hard for lines is polynomial-time decidable for the ring. 
Finally, we show that outside the line and ring environment the problem becomes hard. For a graph as simple as a star, already for the case of two robots, the exploration problem turns out to be NP-complete.

Because of the space constraints, all proofs and some illustrations are moved to the Appendix.

\section{Single Robot on the Line}\label{sct:1onLine}

In this section, we present algorithms that allow a single robot to solve the optimization problem on the line for two cases: when the robots' initial positions are given by an adversary, and when we have the possibility of choosing them ourselves.

We have a sequence of nodes $v_0 < v_1 < \cdots < v_{n-1}$ on the real line, and a robot $r$ initially placed at initial position $r(0)$. We denote by $v_s$ the starting node of the robot, i.e. $r(0) = v_s$. 

\begin{observation}
\label{obs-nonDecreasingDelta} Without loss of generality we may assume that $\Delta_{s+1} < \Delta_{s+2} < \cdots < \Delta_{n-1}$. Indeed, if $\Delta_k \geq \Delta_{k+1}$ for some $k > s$ we can drop node $v_k$ from consideration, since visiting $v_{k+1}$ before its deadline implies that $v_k$ is also visited before its deadline. For the same reason, we can also assume that $\Delta_0 > \Delta_1 > \cdots > \Delta_{s-1}$. 
\end{observation}

\begin{observation}
\label{obs-increasingSequence} Without loss of generality we may consider only the solutions which consist of sequences that are increasing and decreasing at alternate nodes, respectively, i.e., sequences $r(0), r(t_1), r(t_2), \ldots  , r(t_p)$ such that $ 0 \leq r(t_{2i}) < r(t_{2i+2})$, and $ 0 \geq r(t_{2i+1}) > r(t_{2i+3})$, for all $i$ in the appropriate range. Moreover, each turning node $r(t_i)$ is located at some node $v_j , j = 0,1,  \ldots , n-1$.
\end{observation}

\subsection{The \graph}
\label{section:SnapshotGraph}

With these observations in mind, we define the fundamental concept of a directed, layered {\em \graph} $S$~which will form the basis of all subsequent algorithms.


Every node of the  \graph ~$S$~represents a situation when a new node of the line is visited by the robot for the first time. Consequently, each node of $S$~is denoted by a pair $(i,\bar {j})$ or $(\bar {i}, j)$, where $i \leq j$, $[i,j]$ is the interval of nodes already explored by the robot and the node of the line marked with the bar (either $\bar{i}$ or $\bar{j}$) denotes the current position of the robot.

Observe that the robot can advance its exploration in one of two ways: either by visiting the next unexplored node to the left of the interval already explored, or by visiting the next unexplored node to its right. These two possibilities generate the directed edges between the nodes of  the \graph. The weight of such an edge equals the time needed by the robot to traverse the path between robot positions in both nodes. Consequently, the nodes $(i,\bar {j})$ and $(\bar {i}, j)$ are placed at layer $j-i$ and the adjacencies in $S$~are only between nodes of consecutive layers. Notice the following properties of the \graph~(see also Fig.~\ref{fig:snapshot} in the Appendix):
\begin{itemize}
\item The graph $S$~has $n$  layers numbered from $0$ to $n-1$. 
\item There are $n$ {\em source} nodes at the zeroth layer and $2(n-j)$ nodes at  the $j$-th layer for each $j=1,2, \cdots, n-1$. Consequently, there are 2 {\em target} nodes (on the  $(n-1)$-th target layer).
\item The in-degree and the out-degree of each node is bounded by 2. Hence the complexity of the \graph~is $O(n^2)$.
\end{itemize}

Observe that, the solution to the optimization problem for the line corresponds to the shortest path from the source node representing the initial position of the robot to one of its target nodes, which respects the time constraints of all the nodes of $L$.

\subsection{Given initial position of the robot}
\label{Given initial position:subsec}


We first present an algorithm which produces the optimal exploration path, assuming a given starting position $v_s$ of the robot on the line. Consider the \graph~$S$~described above. In order to obtain the optimal exploration path in the \graph~respecting the time constraints of $L$, we generate an all-targets shortest-time tree $T$ whose root  coincides with the  node $(v_s, \bar{v_s})$ of the \graph~corresponding to the initial position $v_s$ of the robot. This is done in the following way.

We add a time counter $time$ to every node of $S$. We set to zero the time counter of the initial node $(v_s, \bar{v_s})$ and to $\infty$ the initial time counters of all other nodes of $S$. We then visit all nodes of $S$~layer by layer. Consider a visit of any such node $v$, which corresponds to the first visit to node $v_j$ of $L$. For each predecessor of $v$ in $S$~ we consider the time equaling its time counter augmented by the weight of the edge joining it with $v$. Let $Min$ denote the smaller of these values (we take an arbitrary one in the case of equality). If $Min$ does not exceed the time constraint of $v_j$ (i.e. $Min \leq \Delta_j$) we set the time constraint of $v$ to $Min$ and we add to $T$ the edge from the corresponding predecessor of $v$. Otherwise, the time counter of $v$ is set to $\infty$ and we leave $v$ parentless.

Observe that, $T$ is a tree, as each node has at most one parent. One of the two target nodes of the smaller time counter defines the optimal exploration time and the path to it in $T$ corresponds to an optimal exploration path of $L$. Otherwise, there exists no exploration path respecting the node deadlines of the line graph.

For any node $v$ of $S$~we denote by $new(v)$ the index of the node of the line \cG~which is newly explored when arriving at $v$. More exactly, $new(v) = j$, such that either $v = (i,\bar{j})$ or $v = (\bar{j},k)$, for some $i \leq j \leq k \leq n-1$.

The following procedure InitStart indicates how to initialize the time counters of the nodes of $S$~before running the main body of the algorithm. For each node $i$ of the line $L$, which may be a starting position of a robot, we put a node $(i, \bar i)$ of $S$ to the set $A$. All nodes of $A$ have their time counters initialized to 0. 

\begin{procedure} [!ht] 
\caption{InitStart($A$, $S$) with $A$ a subset of nodes of $S$ at zeroth layer\;}
\For{every node $v$ of $V(S) \setminus A$}{
	$time(v) = \infty$\;
}
\For{every node $v$ of $A$}{
	$time(v) = 0$\;
}
\end{procedure}

Algorithm~\ref{algo:1RobotLine} describes pseudo-code that formalizes the previously outlined construction of a \tree. 

\begin{algorithm} [!ht] 
\caption{Single Robot exploration on the line with given initial position $v_s$; }\label{algo:1RobotLine}
\KwIn{A snapshot graph $S$ and the starting position $v_s$ of the robot}
\KwOut{An exploration tree with optimal exploration times}
InitStart$(\{v_s\}, S)$\;
\For{layer $i=0$ to $n-1$}{
	\For{each arc $v \to w$ starting at layer $i$}{
		$t =$ time$(v) + weight(v,w)$\;
		\If{$t < time(w) ~and~ t \leq \Delta_{new(w)}$}{
			$time(w) = t$; $v = parent(w)$;		
		}
	}
}
\end{algorithm}

Please see the Appendix for an execution of Algorithm~\ref{algo:1RobotLine}.




\begin{theorem}
\label{1fixed source}
Consider a line graph \cG~and a robot placed at its starting position $v_s$. Algorithm~\ref{algo:1RobotLine}
correctly computes an optimal exploration path which satisfies the node deadlines in $O(n^2)$ time.
\end{theorem}



\subsection{Arbitrary starting position}

We now consider a variation of the problem when the choice of the starting position of the robot is left to the user or it is restricted to be chosen from a subset of nodes of the line graph. We will show that Algorithm~\ref{algo:1RobotLine} also works in such a setting. We need, however, to modify the call to procedure InitStart in line 1 of Algorithm~\ref{algo:1RobotLine}, so that its first parameter equals the set of all nodes of the line at which the robot may start. An example of its execution is presented in the Appendix.

Observe that, for any node $w$ of the \graph, the value of $time(w)$, computed by the algorithm, represents now the shortest exploration time ending at $w$ starting from any node of the line graph. $T$ is now a forest with the nodes of $T$, whose time counter remains at $\infty$ isolated in $T$ (having no children or parent in $T$).

\begin{corollary}
\label{1any source}
Let $A$ be the subset of nodes of the line graph which we can choose for the starting position of the robot.  Suppose that the first parameter of the call to procedure InitStart in line 1 of Algorithm~\ref{algo:1RobotLine} $(A)$ equals the set of all nodes from zeroth level of $S$ which correspond to the nodes of $A$. Such version of Algorithm~\ref{algo:1RobotLine}
correctly computes  in $O(n^2)$ time an optimal exploration path of the line graph, which satisfies the node deadlines. Moreover, for any sub-interval $[i,j]$ of the line, the algorithm computes an optimal robot starting position to explore $[i,j]$, the cost (time) of such exploration and the trajectory of the robot.
\end{corollary}

\section{Multiple Robots on the Line}\label{sct:ManyInLine}

In this section we consider line exploration by a collection of $k < n$ mobile robots. As before we study two variants of the time optimization problem. In the first setting, the distinct initial robot positions are given in advance. In the second setting, the initial positions of the robots are arbitrary, i.e. the algorithm identifies the initial placement of the robots, which results in the shortest exploration time respecting the node deadlines. Both variants are solved using versions of dynamic programming. We start with the following observation concerning the movement of the robots\footnote{We remind the reader that all robots move with identical unit speed.}.

\begin{observation}
\label{disjoint}
There exists an optimal exploration solution in which the robots never change their initial order along the line. Moreover, the sub-intervals of the line explored by different robots are mutually disjoint.
\end{observation}

We use the following notation. Suppose that we need to explore an interval $[i,j]$ of the line respecting the deadlines of the nodes of $[i,j]$. For the setting when the robots are placed at given initial positions, for any pair of indices $i,j$, such that $0\leq i\leq j \leq n-1$, we denote by $T_{i,j}$ the optimal time of exploration of the interval $[i,j]$ using the robots placed within $[i,j]$. When the initial placement of the robots is left to the algorithm, for any $1 \leq r \leq k$, we denote by $T_{i,j}^{(r)}$ the optimal time of exploration of the interval $[i,j]$ using $r$ robots which may be placed at arbitrary initial positions within $[i,j]$.

\subsection{Given initial positions}

%


We start with the following observation

\begin{observation}
\label{start-dyn-prog}
Consider a line and a robot initially placed in its sub-interval $[i,j]$. Using Algorithm~\ref{algo:1RobotLine} the values  $T_{i,j}$ for all $0\leq i\leq j \leq n-1$, may be computed by the formula
\begin{equation}
\label{tij1}
T_{i,j}= \min(time((i, \overline{j})), time((\overline{i},j)))
\end{equation}
\end{observation}

Let $p_i$ denote the initial position of robot $i$. We assume that we have  $0 \leq p_1 < p_2 < \dots < p_k \leq n-1$. By Observation~\ref{disjoint} we need to partition the line into sub-intervals $[l_i, r_i]$ for $i=1,2, \dots, k$ (with $l_1=1$ and $r_k=n$), each one explored by a different robot. The interval $[l_i, r_i]$, explored by robot $i$, contains its initial position $p_i$, but not an initial position of any other robot.  Hence edges $(r_i,l_{i+1})$ for $i=1, \ldots, k-1$, that we call {\em idle edges}, are never traversed by any robot. The following formula, is an obvious consequence of Observation~\ref{disjoint}, 
\begin{equation}
\label{tij}
T_{i,j} = \min\limits_{p_q<m\leq p_{q+1}} \max (T_{i,m-1} ,T_{m,j} ) ,
\end{equation}
for any $i \leq p_{q}, p_{q+1} < j$.
Indeed, 
the idle edge $(m-1,m)$, separating the sub-segments of operation of robots $q$ and $q+1$, is chosen so as to minimize the exploration time of interval $[i,j]$.


We give first an  idea of our algorithm. We generate the \graph, as described in Subsection~\ref{Given initial position:subsec}. Let's use the notation $p_0=-1$ and $p_{k+1}=n$. For $m=1, \dots, k$ let $S_m$ be the subgraph of $S$ obtained by keeping the nodes $(\bar{i},j)$ and $(i, \bar{j})$ such that $p_{m-1} < i,j< p_{m+1}$.
In the first part of our algorithm, for each robot $m$, we run Algorithm~\ref{algo:1RobotLine} with inputs ${p_m}$ and $S_m$, obtaining the optimal exploration time $T_{i,j}$ of each line sub-interval $[i,j]$, which contains exactly one
starting position $p_i$, for $i=1,2,\ldots, k$. 

In the second part of the algorithm, we combine exploration times of individual robots, in order to obtain the optimal exploration time $T_{0,j}$ using robots initially placed within $[0,j]$, subsequently for each $j$. Let $r_j$ denote the number of robots initially placed in interval $[0,j]$ and suppose, that we computed the optimal exploration times of all intervals, which initially contain robots $1,2, \dots, r_j-1$. When $j$ exceeds  $p_{r_j}$ we use robot $r_j$ and we determine the idle edges preceding the intervals of operation of $r_j$, resulting in the optimal exploration times of intervals, which initially contain robots $1,2, \dots, r_j$. 
The formal algorithm (Algorithm~\ref{algo:kRobotsLineFixed}) can be found in the Appendix.

\begin{theorem}
\label{thm:kRobotsLineFixed}
Algorithm~\ref{algo:kRobotsLineFixed} in $O(n^2)$ time computes the optimal exploration of the line by $k$ robots initially placed at given initial positions $0 \leq p_1 < p_2 < \dots < p_k \leq n-1$. 
\end{theorem}

\subsection{Arbitrary initial positions}

This  algorithm is also based on the dynamic programming approach for computing the table $T_{i,j}^{(r)}$, for all $1 \leq r\leq k$ and $0 \leq i<j \leq n-1$. The values of $T_{0,n-1}^{(k)}$ represent the optimal exploration time of the line using $k$ robots. We use the following formula, which works for any $r, r_1, r_2$, where $r_1,r_2 \geq 1$, $r= r_1+r_2$ and any $0 \leq i<j \leq n-1$.

\begin{equation}
\label{tijr}
T_{i,j}^{(r)} = \min_{i\leq k\leq j} \max \left( T_{i,k}^{(r_1)} ,T_{k+1}^{(r_2)} \right) .
\end{equation}

Using Formula~\eqref{tijr}, the values of $T_{i,j}^{(r)}$ may be computed in a greedy manner for the increasing values of $r$. As Formula~\eqref{tijr}~may be naturally computed in $O(n)$ time, the total complexity of such a greedy approach is in $O(kn^3)$. 


We give now a more efficient algorithm computing $T_{0,n-1}^{(k)}$. Observe first, that when $[i_1,j_1] \subseteq [i_2,j_2]$, then $T_{i_1,j_1}^{(r)} \leq T_{i_2,j_2}^{(r)} $. Consequently, when computing $T_{i,j}^{(r)}$, the value of index $k$ which minimizes $\max (T_{i,k}^{(r-1)} ,T_{i,k+1}^{(1)} )$ may be found by a binary search (cf. function OptTime in the Appendix). 

The following observation is easy.
\begin{observation}
\label{obs:logn}
Consider two fixed numbers $r_1, r_2$ of robots. If for any interval $[i,j]$ of the line, $T_{i,j}^{(r_1)}$ and $T_{i,j}^{(r_2)}$ represent the optimal time of exploration of the interval by $r_1$ and $r_2$ robots, respectively, then function OptTime correctly computes in $O(\log n)$ time the optimal exploration time  $T_{i,j}^{(r)}$ of the interval $[i,j]$ by $r=r_1+r_2$ robots.
\end{observation}

The greedy approach would compute the values of table $T_{i,j}^{(r)} $ for any given $r$. Our algorithm below computes the values of $T_{i,j}^{(r)} $ when $r$ is a power of 2 not exceeding $k$. Then, using formula \ref{tijr}, they are combined in $\lceil \log k \rceil$ steps, to compute the values of $T_{i,j}^{(k)}$. 
The formal algorithm (Algorithm~\ref{algo:ArbInit}) can be found in the Appendix.

The following theorem proves the correctness and the complexity of Algorithm~\ref{algo:ArbInit}.

\begin{theorem}
\label{thm:ArbInit}
Algorithm~\ref{algo:ArbInit} computes in $O(n^2 \log n \log k)$ time the optimal time needed by $k$ robots to explore the line.
\end{theorem}


\section{Line Exploration with Unreliable Collections of Robots}\label{sct:LineFaulty}

In this section we study the exploration problem when some of the robots may be faulty, i.e., when they fail to realize their exploration tasks. In this case, other robots need to help, so that eventually every node of the line is visited by some reliable robot before its deadline. Let there be given a weighted line $L$, containing $n$ nodes with given deadlines and a collection of $k$ robots at most $f$ of which may turn out to be faulty. Consider a schedule for $k$ robots on the line $L$. We say that the schedule is $f$-reliable in time $\Delta$, if for any choice of $f$ faulty robots by an adversary, each node of the line is visited by at least one non-faulty robot before its deadline and before time $\Delta$.  

Note that in the case of the presence of unreliable robots, it might be useful to initially place more than one robot at the same position. Consequently, we will assume that it is admissible for more than one robot to start from the same node of the line.

\begin{observation}
\label{f+1}
If there can be $f$ faulty robots, then to successfully explore a node $v$ with deadline $\Delta(v)$, node $v$ must be visited by at least $f+1$ robots before time $\Delta(v)$. 
\end{observation}

It is interesting to look at the {\em decision problem} as well as the {\em optimization problem} related to faulty agents. In the decision problem we look for an algorithm, which, given $f$ and $\Delta$, verifies whether there exists an $f$-reliable schedule in time $\Delta$. In the optimization problem, we need an algorithm, which, for any given $f$, finds the minimal time interval $\Delta$, which admits some $f$-reliable schedule in time $\Delta$. Clearly, solving the optimization problem implies a solution to the decision problem and hardness of the decision problem implies hardness of the optimization problem. We are interested in both settings -- for fixed and for arbitrary initial  positions of the robots. As the case of the arbitrary starting positions is easier we discuss this variant first.


We prove the following theorem.

\begin{theorem}
\label{thm:RingFaulty}
Let there be given a weighted line $L$, containing $n$ nodes with given deadlines and a collection of $k$ robots, which may be put at arbitrary starting positions on $L$. For any $ 0 < f < k$ the optimization problem involving up to $f$ faulty robots may be solved in $O\left(n^2 \log n \log \left\lfloor \frac{k}{f+1} \right\rfloor\right)$ time.
\end{theorem}

We now consider the more difficult case of given starting positions.
%
%
%
Contrary to the case studied in the previous section, when the robots are assigned to fixed positions on the line, the existence of faulty robots leads to a problem which turns out to be NP-hard. In fact, the decision problem is hard, even in the case when all individual deadlines may be ignored (they are all larger than $\Delta$), i.e. when the line does not have any node time constraints.

\probleme{Exploration of the Line with Crash Faults (\DEL)}
{A line $L$, a multiset $P$ of $k$ starting positions of robots, a number of faults $f$ and a time interval $\Delta$.}{Is there an exploration strategy for the collection of $k$ robots, which may include up to $f$ faulty ones, such that 
each node of $L$ is visited by at least one non-faulty robot before time $\Delta$~?}

We construct a polynomial-time many to one reduction from the Numerical 3-Dimensional Matching problem (\fpartb) which is a strongly NP-hard problem (referenced as [SP16] in \cite{johnson1985np}).

\begin{theorem}\label{th:NP-DEL}
The $\DEL$ decision problem is strongly NP-complete.  
\end{theorem} 
\section{The Ring Environment}\label{ sct:Ring}

In this section we show that most of the results for the line environment may be adapted to work on the ring. However, the $\DEL$ decision problem turns out to have a polynomial-time solution for the ring.

Suppose that the ring $R$ contains nodes $0, 1, 2, \dots, n-1$ in that counterclockwise order around $R$. Then every node $i$ of the ring has a counterclockwise neighbour $(i+1) \mod n$ and a clockwise neighbour $(i-1) \mod n$. Consequently, in this section, all the ring node indices are implicitly taken modulo $n$. The approach used for the ring also starts by creating the snapshot graph, however slightly different from the one introduced in Section~\ref{section:SnapshotGraph}. The nodes of the snapshot graph are of the form $(i,\bar {j})$ and $(\bar {i}, j)$, where the node of the ring marked with the bar denotes the current position of the robot and  $[i,j]$ is the segment of the ring already explored by the robot taken in the counterclockwise direction from $i$ to $j$. Observe that, the terminal nodes of the snapshot graph, i.e. those which correspond to the exploration of every node of the ring, are now all nodes $(i,\bar {j})$ and $(\bar {i}, j)$, such that $(j-i) \mod n = 1$, i.e. $i$ is the counterclockwise neighbour of $j$. Such snapshot graph also has $O(n^2)$ nodes of constant degree (see Fig.~\ref{fig:snapshot-graph-ring} in the Appendix). Consequently, by using the argument from Theorem~\ref{thm:kRobotsLineFixed}~we have the following Observation.

\begin{observation}
\label{cor:RingFixedTij}
All values of $T_{i,j}$ for pairs $(i,j)$, such that each pair denotes a counterclockwise segment around the ring containing an initial position of at most one robot, may be computed in amortized $O(n^2)$ time.
\end{observation}

Observe that, there exists an optimal solution for the ring with idle edges between initial positions of consecutive robots. By removing one such edge the ring becomes a line-segment. Consequently, most of our observations for lines may be applied for rings. In particular, for the case of robots which may be placed at arbitrary initial positions on the ring, the following Corollary is obvious.

\begin{corollary}
\label{cor:RingArb}
In $O(n^2 \log n \log k)$ time it is possible to compute the optimal time of exploration of the ring of size $n$ by a set of $k$ robots, which may be placed at arbitrary initial positions.
\end{corollary}

Indeed, it is sufficient to apply Algorithm~\ref{algo:ArbInit}, in which in lines 5 and 12 we consider all pairs $(i,j)$ (rather than pairs for which $i < j$). 

In the case of robots at given initial positions, the adaptation of the line algorithm to the ring case is also relatively easy, with some compromise on its time complexity.
We have the following Proposition.
\begin{proposition}
\label{prop:RingFixedTij}
There exists an $O\left(n^2+\frac{n^2}{k} \log n \right)$ algorithm for computing an optimal exploration of the ring $R$ of size $n$ using $k$ mobile robots, initially placed at fixed positions on $R$.
\end{proposition}
%
We now consider unreliable robots.
Similarly to the line exploration case, every node of the environment must be explored $f+1$ times by different robots before its deadline. 

Consider first the case of robots which may be placed at arbitrary initial positions on the ring $R$. Suppose that we denote by $R^{(f+1)}$ a ring obtained in the following way. We cut $R$ at any node $v$, obtaining a line segment starting and ending by a copy of $v$. We merge $f+1$ copies of such segment, identifying the starting and the ending nodes of consecutive copies, obtaining a segment of $n(f+1)$ nodes. Finally, we identify both endpoints of such segment obtaining a ring $R^{(f+1)}$. Observe that, covering $R$ by $k$ robots' exploration trajectories, so that each node of $R$ is visited $f+1$ times, is equivalent to exploring $R^{(f+1)}$ using $k$ robots, so that each of its nodes is visited (once) before its deadline. As the size of $R^{(f+1)}$ is in $O(nf)$, from Corollary~\ref{cor:RingArb} we get 

\begin{corollary}
\label{cor:RingArbFaulty}
Suppose that in an $n$-node ring we can place at arbitrary initial positions $k$ robots, which may include up to $f$ faulty ones.  In $O(n^2f^2 \log k(\log n + \log f))$ time it is possible to compute the optimal time of exploration of the ring.
\end{corollary}

If the initial positions of the robots on the ring are given in advance, contrary to the case of the line segment, it is possible to decide in polynomial time whether there exists an $f$-reliable schedule in any given time $\Delta$.

\begin{proposition}
\label{prop:RingFixedTij1}
Consider a ring $R$ of size $n$ and $k$ robots placed at given initial positions at the nodes of $S$. For any given time $\Delta$ it is possible to decide in polynomial time whether ring $R$ may be explored by its robots within time $\Delta$.


\end{proposition}

\section{NP-hardness for Star Graphs}

We gave exploration algorithms for lines and rings with time constraints on the nodes. It is easy to see that the exploration problem is hard for graphs, even for the case of a single robot and a graph with edges of unit length. Indeed, for a graph on $n$ nodes, by setting all its node deadlines to $n-1$, an instance of exploration problem is equivalent to finding a Hamiltonian path. However, we show below that the exploration problem is hard for graphs as simple as stars and already for two mobile robots. 
We construct a polynomial-time reduction from the Partition Problem~\cite{garey2002computers}.

\begin{proposition}\label{prop:star}
The exploration problem respecting node deadlines for given starting positions of the robots  is NP-hard. This problem is also NP-hard if the starting positions are arbitrary.
\end{proposition}

\section{Conclusion and Open Problems}

We studied the question of exploring graphs with time constraints by collections of unreliable robots. When all robots are reliable we used dynamic programming to give efficient exploration algorithms for line graphs and rings. We showed, however, that the problem is NP-hard for graphs as simple as stars. We showed how to extend, in most cases, our solutions to unreliable collections of robots. One of our results is quite unexpected and important. Suppose that a collection of robots, placed on a line, may contain an unknown subset of robots (of bounded size), which turn out to be crash faulty. Verifying whether it is possible to explore the line within a given time bound is an NP-hard problem. The same problem on the ring has a polynomial-time solution.

An interested reader  may observe that our positive results imply the possibility to compute the {\em resilience} of the configuration, i.e. given a time $\Delta$, to find the largest value $f$, such that there exists a schedule assuring exploration when any set of $f$ robots turns out to be unreliable.

In our paper, we did not actually produce schedules for our robots, but we only computed the optimal times when such schedules may be completed. However, from our work it is implicitly clear how to generate such schedules. We proved the optimality of the schedules but we did not prove the optimality of our algorithms. One of the possible open problems is to attempt to design algorithms of better time complexity. 

Several other open problems remain, especially those related to feasibility questions. The exploration of the star graph by a single robot may be done in polynomial time. For example, an interested reader may observe, that it is sufficient to visit each star node $v_i$ of degree one in increasing order of the value $\Delta_i+w_i$, where $\Delta_i$ and $w_i$ are, the deadline of node $v_i$ and the weight of its incident edge, respectively. It is possible to show that, either such schedule is feasible, or there does not exist any feasible schedule. Is it possible to extend the algorithm for single robot in a star for some larger class of graphs? How is a tree with node deadlines explored by a single robot? Or, conversely, what is the smallest (or simplest) class of graphs for which the exploration by a single robot is hard?

%
%

\newpage
\bibliographystyle{plain}
\bibliography{refs}
\thispagestyle{empty}

\newpage
\appendix

\pagenumbering{roman}
\section*{Appendix}

\section{Illustration of a snapshot graph}
\begin{figure} [!htb]
\begin{center}
\includegraphics[width=0.75\textwidth]{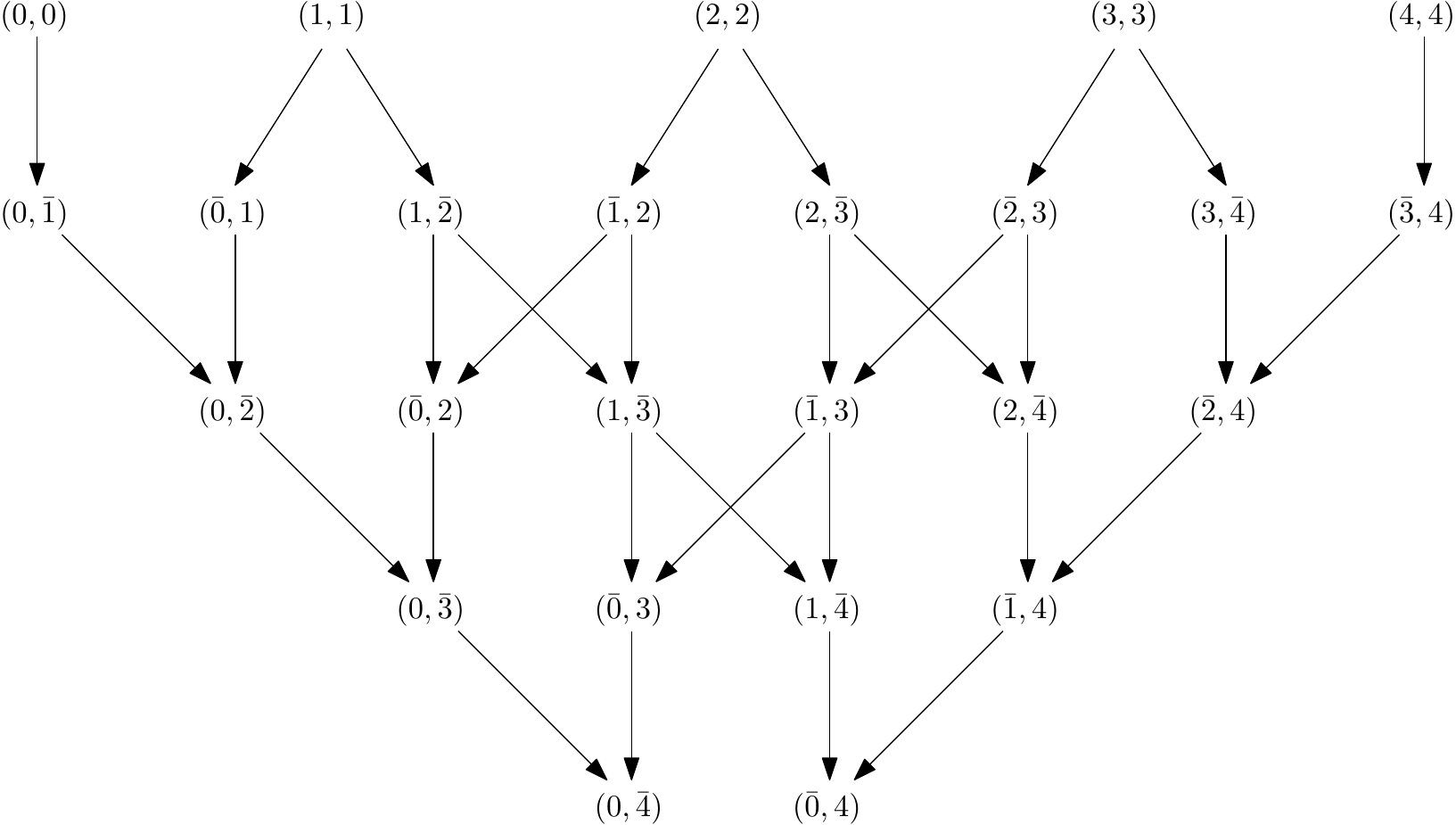}
\end{center}
\caption{A depiction of a \graph~for the case of line $L$ consisting of five nodes.}
\label{fig:snapshot}
\end{figure}

For clarity, we do not show the edge weights of the \graph ~$S$~ in Fig.~\ref{fig:snapshot}. Notice that, for any line graph $L$, the weights of the directed edges $(i,\bar {j})\rightarrow (i,\overline {j+1})$ and  $(\bar{j},k)\rightarrow (\overline{j+1},k)$ in its \graph~are equal to the weight of the edge $(j,j+1)$ in the line graph $L$. Similarly, the weights of the directed edges  $(i-1,\bar {j})\rightarrow (\bar {i},j)$ and  $(\bar {i},j-1)\rightarrow (i,\bar {j})$  in the \graph~$S$~are equal to the weight of the path $i  \rightsquigarrow j$ in the line graph $L$.

\section{Illustration of the execution of Algorithm~\ref{algo:1RobotLine} for a given initial position of the robot}
Figure \ref{fig:snapshot-tree} illustrates the execution of Algorithm~\ref{algo:1RobotLine}. The weighted line graph containing five nodes denoted by integers from $0$ to $4$ is presented at the top of Fig. \ref{fig:snapshot-tree}. The robot is initially placed at node 1. The solid directed edges depict the \tree~respecting the deadlines (the remaining edges of the \graph~which are not being used are dashed). Each node has been assigned the time counter computed by Algorithm~\ref{algo:1RobotLine}.
The path of the \tree~ending in the target node represents the optimal trajectory of the robot. 

\begin{figure} [!htb]
\begin{center}
\includegraphics[width=0.75\textwidth]{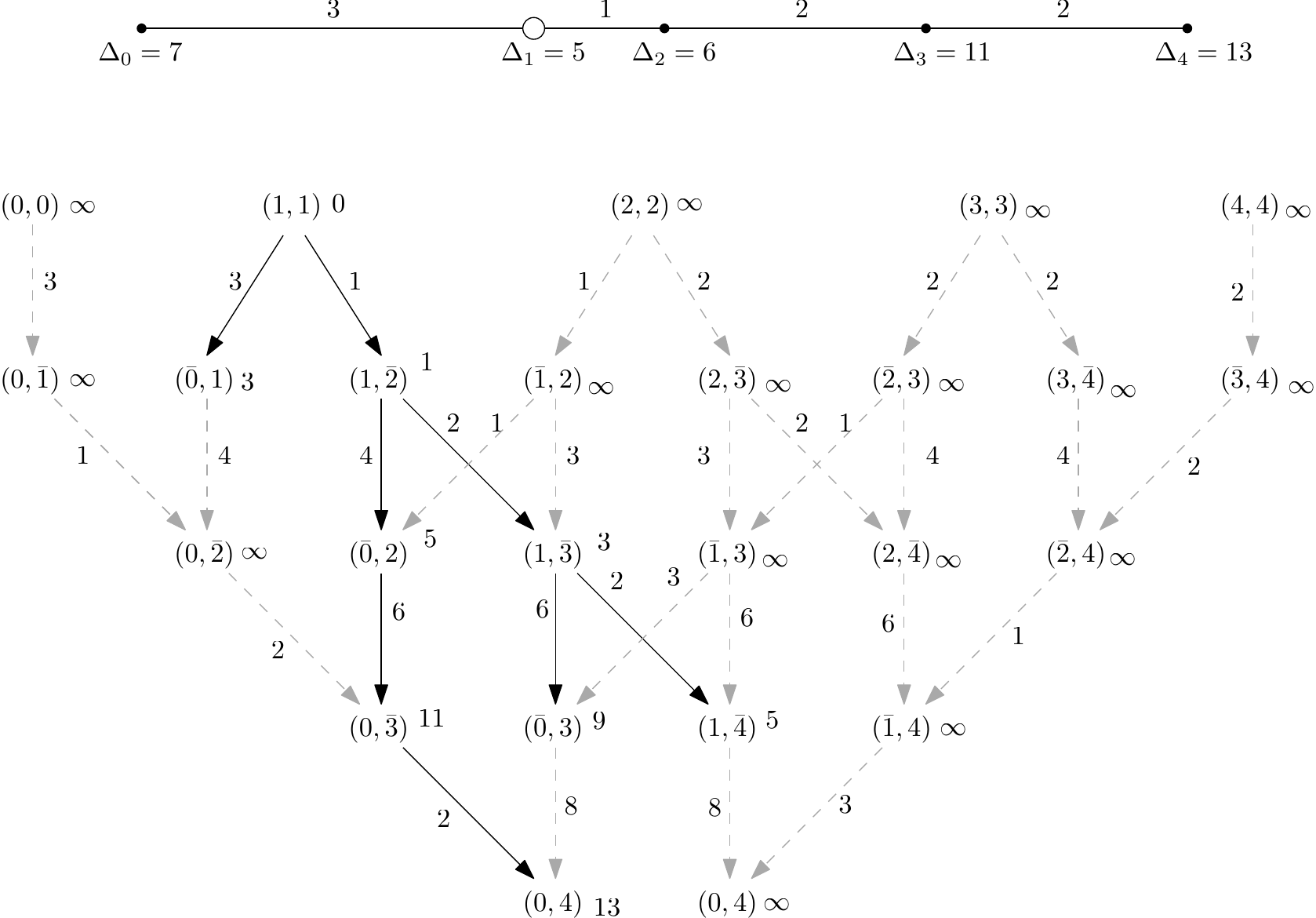}
\end{center}
\caption{Illustration of the execution of the line-exploration algorithm starting from node 1.}
\label{fig:snapshot-tree}
\end{figure}

\section{Proof of Theorem~\ref{1fixed source}}
\begin{proof} 
We show that for every node $v=(\bar{i},j)$ (resp. $v=(i, \bar{j})$) of the \graph~the algorithm computes the shortest time $time (v)$ needed to explore the interval $[i,j]$ of the line graph.which respects the deadlines of its nodes by the robot starting at $v_s$, such that $i \leq s \leq j$, and ending its exploration at $v_i$ (resp. $v_j$). The proof goes by induction on the layer. The claim is clearly true for any node $(i, \bar{i})$ at layer $0$. Suppose that the claim is true for all nodes at layers at most $\ell-1$. Take any node at level $\ell$, i.e., either $v= (\bar{i}, i+\ell)$ or $v=(i, \overline{i+\ell})$ . Consider the shortest time exploration path ending at $v$. The immediate predecessors of $v$ in this path is a node $w$ from layer $\ell -1$, for which the shortest-time exploration path is correctly computed by the inductive hypothesis.  The time needed to reach $v$ from $w$ equals the time distance between $new (v)$ and $new (w)$ at the line graph. If $time (v) + weight (w \to v) \leq \Delta_{new(w)}$ then the deadline of node $new (w)$ is respected and $time (w)$ is correctly computed lines $4$-$6$ of Algorithm~\ref{algo:1RobotLine}, otherwise the exploration time of $w$ remains at $time (w) = \infty$ as set in the InitStart procedure. The $O(n^2)$ time complexity follows directly from the properties of the \graph.  
\qed
\end{proof}

\section{Illustration of the execution of the modified Algorithm~\ref{algo:1RobotLine} for an arbitrary starting position}
An example of the execution of Algorithm~\ref{algo:1RobotLine} (modified as described in the subsection) for an arbitrary starting position is presented on Fig.~\ref{fig:snapshot-tree-any-source}, where a user may choose any node of the line graph as the starting position of the robot.

\begin{figure} [!htb]
\begin{center}
\includegraphics[width=0.75\textwidth]{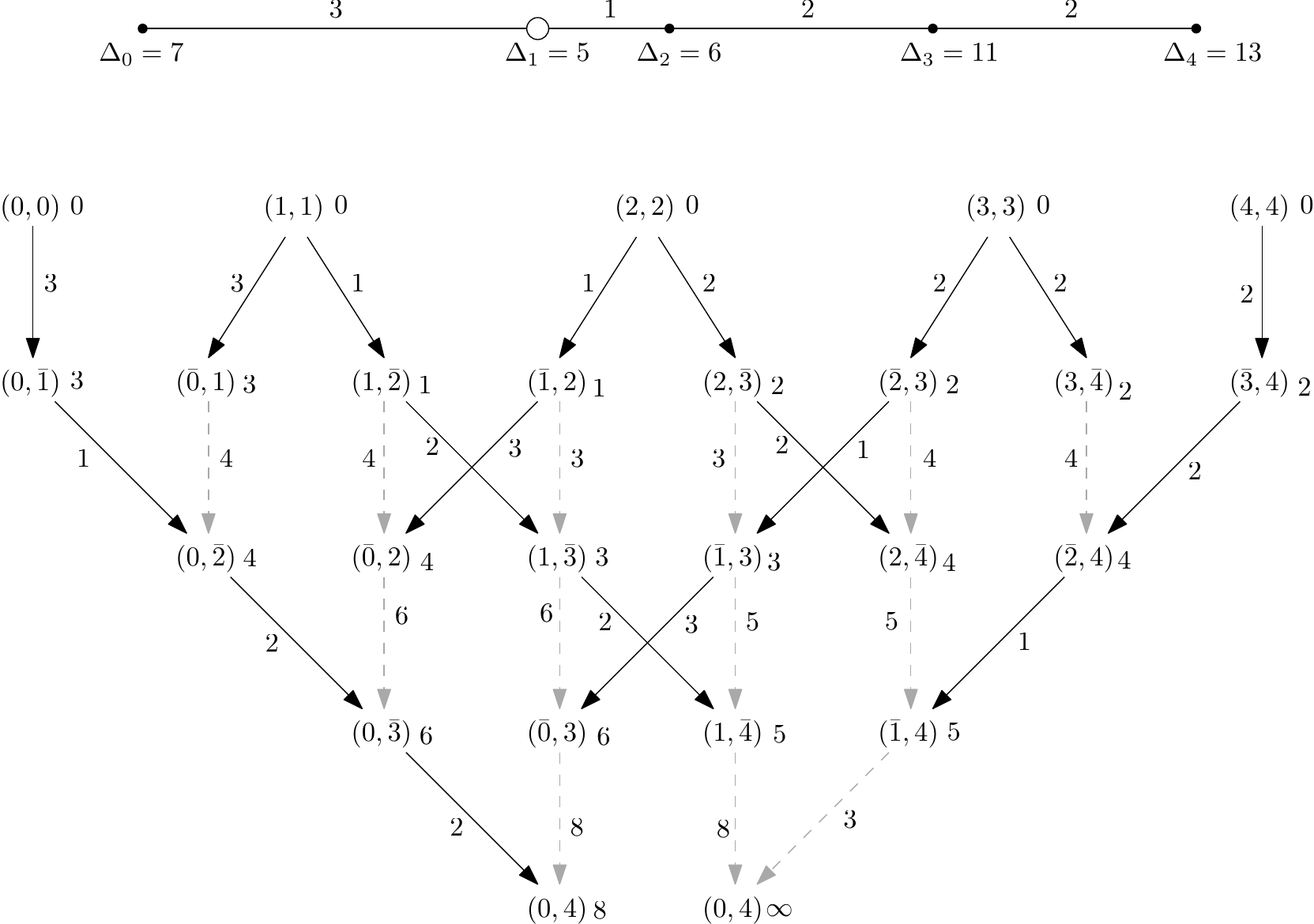}
\end{center}
\caption{Illustration of the execution of the line-exploration algorithm in the case of arbitrary starting  node. For any sub-interval $[i,j]$ of the line, the optimal robot trajectory exploring $[i,j]$ is given by the cheaper among the solid directed paths incoming to nodes $(i, \bar j)$ and $(\bar i, j)$.}
\label{fig:snapshot-tree-any-source}
\end{figure}

\section{Proof of Corollary~\ref{1any source}}
The proof of Corollary~\ref{1any source} is almost identical to the case of Theorem~\ref{1fixed source}. Observe that, in the inductive step, when the parent $w$ of any node $v$ in $T$ is determined, the root of the connected component of $T$ containing $v$ corresponds to the node of $A$ offering the shortest exploration path.

\section{Proof of Theorem~\ref{thm:kRobotsLineFixed}}

\begin{proof}
By Observation \ref{disjoint}, we can assume that in an optimal solution, each robot $m$ operates in an interval $[l_m,r_m]$, which does not contain an initial position of any other robot. Hence we have $p_{m-1}< l_m \leq p_m$ and $p_{m} \leq r_m < p_{m+1}$. All pairs of indices $(i,j)$, which verify this property are considered in line 4 of Algorithm \ref{algo:kRobotsLineFixed}.
By Theorem~\ref{1fixed source}~and Observation \ref{start-dyn-prog}~each such value $T_{i,j}$ is correctly computed in line 5 of Algorithm \ref{algo:kRobotsLineFixed}.

We now prove that in line 7 the Algorithm \ref{algo:kRobotsLineFixed} correctly computes values  $T_{0,j}$ for  all $j=0,1, \dots, n-1$. The proof goes by induction on $j$. For all $0 \leq j  < p_2$ the interval  $[0,j]$ contains a single robot, so the value $T_{i,j}$ is correctly computed in the first iteration of the for-loop in lines 2-5. Consider any $j \geq p_2$, i.e. when the interval $[0,j]$ contains more than one robot, and suppose, by the inductive hypothesis, that the values $T_{0, i}$ correspond to optimal times of exploration of segments $[0,i]$, for all $i<j$ . Let $T^*$ be the optimal time of exploration of interval $[0,j]$, which verifies the claim of Observation \ref{disjoint}, i.e. such that there exists an idle edge $(m^*-1, m^*)$, and $p_{r_{j-1}}< m^* \leq p_{r_j}$. During such optimal exploration, robots $1,2, \dots, r_j-1$ explore interval $[0, m^*-1]$ (using some time $T_1^*$), and robot $r_j$ explores interval $[m^*,j]$ (using time $T_2^*$). Clearly $T^* = \max (T_1^*, T_2^*)$. By the inductive hypothesis, we have $T_{0,m^*} \leq T_1^*$  and $T_{m^*+1,j} \leq T_2^*$. Consequently, we have in line 7 of Algorithm~\ref{algo:kRobotsLineFixed}
\begin{align*}
T_{0,j} &= \min\limits_{p_{r_{j-1}}< m \leq p_{r_j}} \max \, \left(T_{0,m-1}, T_{m,j} \right) 
 \leq  \max \, (T_{0,m^*-1}, T_{m^*,j})
 = \max (T_1^*, T_2^*) 
 = T^*,
\end{align*}
which concludes the inductive proof.

We consider now the time complexity of Algorithm~\ref{algo:kRobotsLineFixed}. The snapshot graph $S$ in line 1 is constructed in $O(n^2)$ time. Observe that since each node of $S$ can only be in two different subgraphs $S_i$ and $S_j$, we have $\sum_{i=1}^{k}| V(S_i)|\leq 2|V(S)|= O(n^2)$. Hence, all the executions of line 3 of Algorithm~\ref{algo:kRobotsLineFixed} take $O(n^2)$ amortized time. Similarly, in line 4 of the algorithm, in all its executions, it considers $O(n^2)$ nodes of graph $S$. Consequently the for-loop of lines 2-5 is executed in $O(n^2)$ amortized time. As each of $O(n)$ executions of the for-loop in lines 6-7 takes $O(n)$ time we conclude the $O(n^2)$ overall time complexity of Algorithm~\ref{algo:kRobotsLineFixed}. 
\qed
\end{proof}

\begin{algorithm} [!ht]
\caption{Exploration algorithm on the line with $k$ robots at fixed initial positions\label{algo:kRobotsLineFixed}}
\KwIn{Line $L$ with starting robots' positions $p_1, p_2, \dots p_k$}
Construct the snapshot graph $S$ from $L$\;
\For{$m=1$ to $k$}{
	Execute Algorithm~\ref{algo:1RobotLine} with inputs $p_m$ and $S_m$\;
	\For{every $(i,j)$ s.t. $p_{m-1}< i\leq j< p_{m+1}$}{
		$T_{i,j} := \min \{ t_m(i, \bar{j}) , t_m(\bar{i}, j ) \}$\;
	}
}
\For{$j=p_2$ to $n-1$}{
	$T_{0,j}:= \min\limits_{p_{r_{j-1}}< m \leq p_{r_j}} \max \, (T_{0,m-1}, T_{m,j})$\;
}
\end{algorithm}


\begin{function} [!ht] 
\caption{ OptTime($i,j,r_1, r_2$); }
\If {$j-i+1 \leq r_1+r_2$} {\bf{return} 0}
$k_{low} = i; k_{high} = j$\; 
\While{$k_{low} < k_{high+1} $}{
         $k=(k_{low}+k_{high})/2$\;
	\If {$T_{i,k}^{(r_1)} < T_{k+1,j}^{(r_2)}$}{
			 $k_{low}=k$		
		}
		\Else {$k_{high}=k$}
}
{\bf{return} $\min(  \max (T_{i,k_{low}}^{(r_1)} ,T_{k_{low},j}^{(r_2)} ), \max (T_{i,k_{high}}^{(r_1)} ,T_{k_{high},j}^{(r_2)} ))$;}
\end{function}

\begin{algorithm} [!ht] 
\caption{; Multiple robot line exploration with arbitrary starting positions \label{algo:ArbInit}}
Let $r_b, r_{b-1}, \dots, r_0$ be the consecutive bits of the binary representation of $k$\;
Compute table $T_{i,j}^{(1)} $ \;
\For{$m=0$ to $b$}{
	$r=2^m$\;
	\For{all pairs $(i,j)$ such that $0\leq i <j \leq n-1$}{
	$T_{i,j}^{(2r)} =$ OptTime$(i,j,r, r)$\;
	}
}
$r = 2^{b}$\;
\For{$m=1$ to $b$}{
	\If{ $r_{b-m}=1$}{
		$p=2^{b-m}$\;
	\For{all pairs $(i,j)$ such that $0\leq i <j \leq n-1$}{
		$T_{i,j}^{(p+r)} =$ OptTime$(i,j,p, r)$\;
		$r=p+r$\;
	}
	}
}
\end{algorithm}

\section{Proof of Theorem~\ref{thm:ArbInit}}

\begin{proof} 
By Corollary~\ref{1any source} and Formula~\eqref{tij} the usage of Algorithm~\ref{algo:1RobotLine}~in line 2 of Algorithm~\ref{algo:ArbInit} correctly computes a single robot optimal exploration time for any sub-interval of a given line. By induction on $r$, using Observation~\ref{obs:logn}, lines 3-6 of Algorithm~\ref{algo:ArbInit} correctly compute the optimal exploration time of any interval $[i,j]$ using $2^m$ robots, for any $m$, such that $2^m < r$.  

From line 1  we have $k=r_b2^{b}  + r_{b-1}2^{b-1}+\dots+r_02^0$, where $b$ is the position of the first 1-digit in the binary representation of $k$. We prove that, at the start of each iteration of the {\em for} loop from line 8, we have 
\begin{enumerate}
\item $r=k - k \mod 2^{b+1-m}$, and  
\item the table $T_{i,j}^{(r)}$ has been already computed for all $0\leq i <j \leq n-1$.
\end{enumerate}
The proof goes by induction on $m$. At the start of the first iteration of the loop when $m=1$, we have $r=2^b$. Then indeed the inductive condition is verified as 
$$
k- k \mod 2^{b+1-1} = (r_b2^{b}  + r_{b-1}2^{b-1}+\dots+r_02^0) -  (r_{b-1}2^{b-1}+r_{b-2}2^{b-2}+\dots+r_02^0) = 2^b = r
$$
and the value of $T_{i,j}^{(2^b)}$ was computed previously in line 6 of the algorithm. 

Suppose that the inductive condition was verified at the beginning of the $m$-th iteration. Suppose first that $r_{b-m}=0$. Then the $i$-th iteration of the loop is empty but as $k \mod 2^{b+1-m} = k \mod 2^{b+1-(m+1)}$, so that at the beginning of the next iteration the value of $r$ remains unchanged, it follows that the inductive condition is verified.

Consider now the case when $r_{b-m}=1$. Then, between the start of the $m$-th and the $(m+1)$-st iteration of the loop in lines 10 and 13 we have $r:=r+ 2^{b-m}$. Consequently, by the inductive assumption, we have at the beginning of the $(m+1)$-st iteration

\begin{align*}
r &= k - k \mod 2^{b+1-m} + 2^{b-m}\\
&= (r_b2^{b}  + r_{b-1}2^{b-1}+\dots+r_02^0) -  (r_{b-m}2^{b-m}+\dots+r_02^0) +2^{b-m} \\
&= (r_b2^{b}  + r_{b-1}2^{b-1}+\dots+r_{b-m}2^{b-m})=  k - k \mod 2^{b+1-(m+1)}
\end{align*}
The value of the table $T_{i,j}^{(r)}$ is then computed in line 12 of the algorithm, which completes the induction proof.

From the inductive proof it follows that at the end of the $b$-th iteration of the {\em for} loop from line 8 (i.e. at the beginning of the non-existing $(b+1)$-st iteration) we have $r=k - k \mod 2^{b+1-(b+1)}=k$, and  
the table $T_{i,j}^{(r)}=T_{i,j}^{(k)}$ has been computed, which completes the proof of the correctness of the algorithm.

In line 2, the table $T_{i,j}^{(1)} $ may be computed by Algorithm~\ref{algo:1RobotLine} in $O(n^2)$ time (cf. Fig.~\ref{fig:snapshot-tree-any-source}).
As $r < 2^b$, both {\em for} loops starting at line 3 and 8 have $O(\log k)$ iterations. Since each internal  {\em for} loop from line 5 and 11, respectively, has $O(n^2)$ iterations calling function OptTime of complexity $O(\log n)$ we conclude that Algorithm~\ref{algo:ArbInit} finishes in $O(n^2 \log n \log k)$ time.
\qed
\end{proof}

\section{Proof of Theorem~\ref{thm:RingFaulty}}
\begin{proof}
Let $\Delta$ be the time interval satisfying the claim of the theorem, in the sense that there exists an $f$-reliable schedule in time $\Delta$, while for any $\Delta' <\Delta $, there does not exist an $f$-reliable schedule in time $\Delta'$. 
We show first that the necessary and sufficient condition for the existence of such an $f$-reliable schedule is the following.

\noindent {\bf Condition 1:} {\em There must exist a schedule involving $\left\lfloor \frac{k}{f+1} \right\rfloor$ robots (all reliable) at arbitrary initial positions on $L$, which solves the exploration of $L$ in time $\Delta$.}

Indeed, by Observation \ref{f+1}~each node of the line $L$ must be explored by at least $f+1$ robots. Therefore we can partition the collection of robots into $f+1$ groups, each group entirely exploring line $L$. The least numerous of these groups can contain no more than $\left\lfloor \frac{k}{f+1} \right\rfloor$ robots and this group must explore $L$. Conversely, if $\left\lfloor \frac{k}{f+1} \right\rfloor$ robots can explore line $L$ in time $\Delta$, we can form $f+1$ groups of $\left\lfloor \frac{k}{f+1} \right\rfloor$ robots each, executing the same schedule and the line is explored by each of $f+1$ independent groups.

By Theorem~\ref{thm:ArbInit}, Algorithm~\ref{algo:ArbInit} computes the optimal time of line exploration by a collection of robots which may be placed at arbitrary initial positions. Consequently, the output of Algorithm \ref{algo:ArbInit} run for $r=\lfloor \frac{k}{f+1} \rfloor$ robots exactly verifies Condition 1. By Theorem~\ref{thm:ArbInit}, its time complexity is then as stated in the claim of the theorem.
\qed
\end{proof}
\section{Proof of Theorem~\ref{th:NP-DEL}}

The proof of Theorem \ref{th:NP-DEL} is split into two lemmas. We first show that the $\DEL$ decision problem is strongly NP-hard, and then that the $\DEL$ decision problem is in NP.

\begin{lemma}\label{lem:NP-hard-DEL}
The $\DEL$ decision problem is strongly NP-hard.  
\end{lemma}

\begin{proof}
We construct a polynomial-time many to one reduction from the following strongly NP-hard problem referenced as [SP16] in \cite{johnson1985np}.

\probleme{Numerical 3-Dimensional Matching (\fpartb)}
{Three multisets of positive integers $A = \{a_1, a_2, \dots, a_q\}, B = \{b_1, b_2, \dots, b_q\}, C = \{c_1, c_2, \dots, c_q\}$, and an integer $S$.}
{Does there exist two permutations $\pi_B,\pi_C$ of $[1,q]$ such that for every $1\leq i\leq q$, $a_i+b_{\pi_B(i)}+c_{\pi_C(i)} =S$?}

We construct an instance $(L,P,f, \Delta)$ of the $\DEL$ problem from an instance of $\fpartb$ as follows. Let $a= \max_{i\in [1,q]}(a_i)$, $b= \max_{i\in [1,q]}(b_i)$ and $c= \max_{i\in [1,q]}(c_i)$. Let $I=4S+6a+6b+12c$ and $\ell=3I-4S-1$. $L$ is the line of length $\ell$ (with $\ell+1$ nodes).
Each edge of $L$ has weight one. For the sake of simplicity we name $i$ the node of $L$ at distance $i$ from the leftmost node. We have $3q$ robots each corresponding to an integer from one of the multisets $A, B$ or $C$. For each $i=1,2,\dots, q$, 
we put three robots: one robot $\cA_i$ at node $\alpha_i=a_i$, one robot $\cB_i$ at node $\beta_i=I+2b_i$ and one robot $\cC_i$ at node $\gamma_i=2I+4c_i$. The number of faults $f$ is equal to $q-1$ and the time interval $\Delta$ is equal to $I-1$. The construction can be done in polynomial time. We show that the answer to the constructed instance of the $\DEL$ problem is the same as the answer to the original instance of $\fpartb$.

First, assume that there exists a solution $\pi_B,\pi_C$ to the instance of the $\fpartb$ problem. We show that the robots can solve the corresponding instance of the $\DEL$ problem as follows. 


Robot $\cA_{i}$ will first move to the left until reaching node $0$ (moving distance $a_{i}$), and then to the right until reaching node $\alpha_{i}^r=I-1-a_{i}$ (moving distance $I-1-a_{i}$). This can be done in time $\Delta = I-1$ and thus robot $\cA_{i}$ has visited in time all nodes in the interval $[0, \alpha_{i}^r]$. 

Robot $\cB_{\pi_B(i)}$ will first move to the left until reaching node $\beta_{\pi_B(i)}^l =\alpha_{i}^r+1$ (moving distance $a_{i}+ 2b_{\pi_B(i)}$) and then to the right until reaching node $\beta_{\pi_B(i)}^r=2I-1-2a_{i}-2b_{\pi_B(i)}$ (moving distance $I-1- a_{i}- 2b_{\pi_B(i)}$). This can be done in time $\Delta = I-1$ and thus it has visited in time all nodes in the interval $[\beta_{\pi_B(i)}^l, \beta_{\pi_B(i)}^r]$.

Robot $\cC_{\pi_C(i)}$ will first move to the left until reaching node $\gamma_{\pi_C(i)}^l =\beta_{\pi_B(i)}^r+1$ (moving distance $2a_{i}+ 2b_{\pi_B(i)}+ 4c_{\pi_C(i)}$) and then to the right until reaching node $\gamma_{\pi_C(i)}^r=3I-1-4a_{i}-4b_{\pi_B(i)}-4c_{\pi_C(i)}$ (moving distance $I-1-2a_{i}- 2b_{\pi_B(i)}- 4c_{\pi_C(i)}$). This can be done in time $\Delta = I-1$ and thus it has visited in time all nodes in the interval $[\gamma_{\pi_C(i)}^l, \gamma_{\pi_C(i)}^r]$.
 Observe that we have :
\begin{align*}
& 3I-1 -4a_{i}-4b_{\pi_B(i)}-4c_{\pi_C(i)}  =  3I - 4S - 1  =  \ell
\end{align*}

Hence $\cC{\pi_C(i)}$ has visited in time all nodes in the interval $[\delta_{\pi_C(i)}^l, \ell]$.

\begin{figure} [!htb]
\begin{center}
\includegraphics[width=0.75\textwidth]{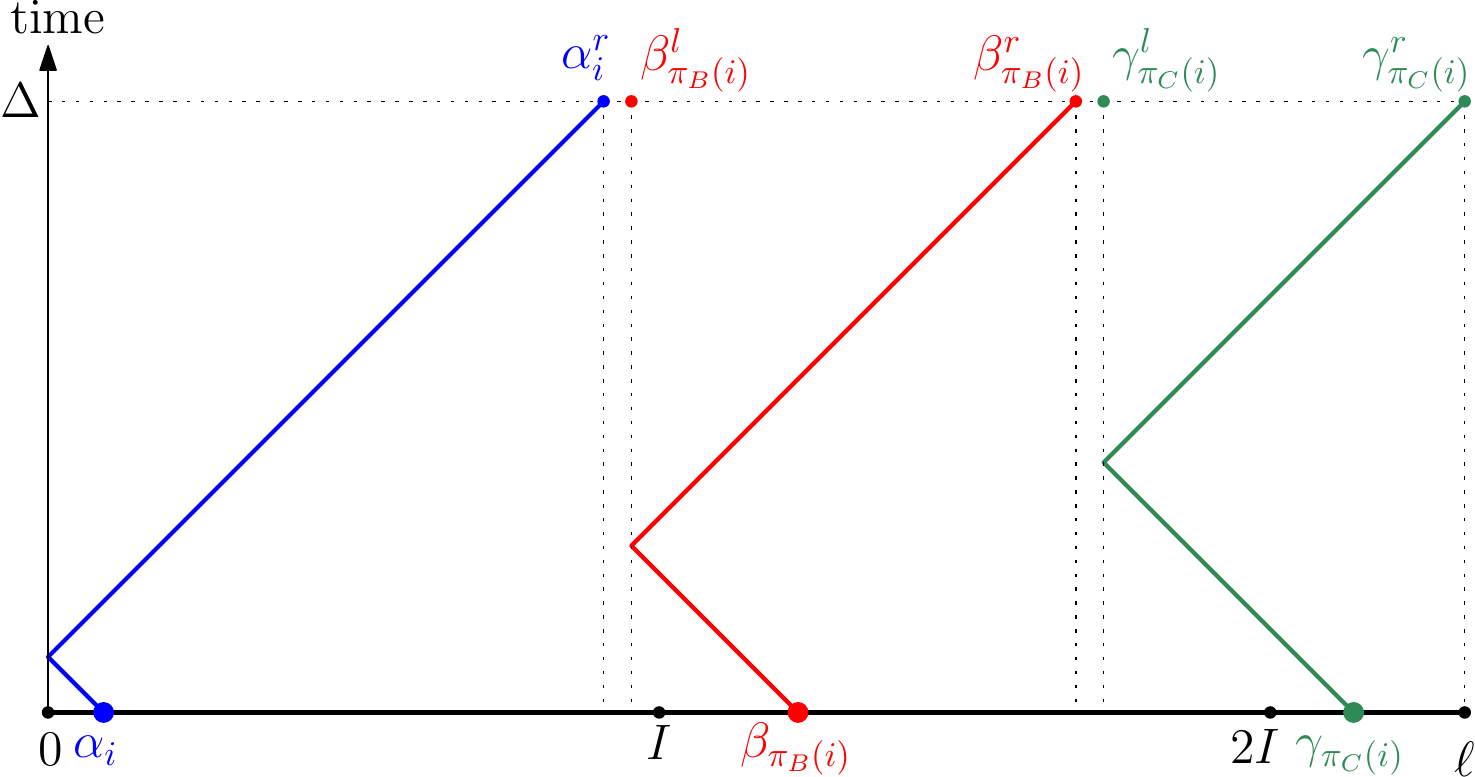}
\end{center}
\caption{An illustration of robots $\cA_i$, $\cB_{\pi_B(i)}$ and $\cC_{\pi_C(i)}$ visiting all nodes of the line.}
\label{fig:covering-quadruple}
\end{figure}

Since $[0, \alpha_{i}^r]\cup [\beta_{\pi_B(i)}^l, \beta_{\pi_B(i)}^r]\cup 
[\gamma_{\pi_C(i)}^l, \ell] =[0,\ell]$, all nodes of the line have been visited in time by either $\cA_i$, $\cB_{\pi_B(i)}$ or $\cC_{\pi_C(i)}$ for $i = 1,2,\dots, f+1$. It follows that every node is visited by at least one non-faulty robot and this is a solution to the $\DEL$ problem.

Now assume there is a solution to the $\DEL$ problem. Let $\bA=\{\cA_i\mid 1\leq i\leq q\}$, $\bB=\{\cB_i\mid 1\leq i\leq q\}$ and $\bC=\{\cC_i\mid 1\leq i\leq q\}$. First, we show the following claim.

\begin{claim}\label{claim:part}
The robots in $\bA$ must visit node $0$ and they are the only robots that can do it,
the robots in $\bB$ must visit node $I$ and they are the only robots that can do it and
the robots in $\bC$ must visit node $\ell$ and they are the only robots that can do it.
\end{claim}

For $i = 1,2,\dots, q$, the robots in positions $\beta_i$ and $\gamma_i$ are to far (at distance at least $I$) to reach node $0$ in time smaller than $\Delta = I-1$. The robots in positions $\alpha_i$ are the only ones that can visit the node $0$ and since this node must be visited by $f+1=q$ robots, they must all visit it. Hence, the robots in positions $\alpha_i$ cannot visit in time node $I$ which is at distance $I$ of node $0$. Similarly, the robots in positions $\beta_i$ are the only ones that can visit the node $I$ and since this node must be visited by $f+1=q$ robots, they must all visit it. Similarly, the robots in positions $\gamma_i$ are the only ones that can visit the node $2I$ and so node $\ell$ (since $\ell >2I$) and since this node must be visited by $f+1=q$ robots, they must all visit it. This ends the proof of the claim. 

For $i = 1,2,\dots, q$, let $[0, \alpha_i^r]$ be the interval of nodes visited by robot $\cA_i$,
$[\beta_i^l, \beta_i^r]$ be the interval of nodes visited by robot $\cB_i$ and
$[\gamma_i^l, \gamma_i^r]$ be the interval of nodes visited by robot $\cC_i$. 

\begin{claim}\label{claim:covering}
There are two permutations $\pi_B(i)$ and $\pi_C(i)$ such that 
for $i=1,2\dots,q$, $\beta_{\pi_B(i)}^l = \alpha_{i}^r +1 $ and $\gamma_{\pi_C(i)}^l = \beta_{\pi_B(i)}^r +1$.
\end{claim}

First observe that if there is a portion of the line that is visited by more than $f+1=q$ robots, then it means that there are robots from two different sets (for example, robots from sets $\bA$ and $\bB$). We can then cut the trajectory of some of the robots in order to decrease the number of robots visiting the same node. So we can assume without loss of generality that each node is visited by exactly $q$ robots. This means that there is a partition of the robots into $q$ subsets, such that every node of the line is visited in time by exactly one robot of each subset. By the first claim, there is one robot of each set $\bA$, $\bB$ and $\bC$ in each of $q$ subsets. Hence, for $i=1,2,\dots,q$, there is a subset of robots $\{\cA_i, \cB_{\pi_B(i)}, \cC_{\pi_C(i)}\}$ that must visit all the nodes of the line. Since two robots of the same subset do not visit the same node, their intervals are disjoint. This ends the proof of the claim.

By the second claim, 
Robot $\cA_i$ can travel a distance $\Delta$ to search its interval $[0, \alpha_i^r]$.
Observe that $\cA_i$ starts at distance $a_i$ from node $0$. Since $\Delta > 3a\geq 3a_i$, 
the optimal way for robot $\cA_i$ to search its interval is to first go to the left and 
then to the right. So, we have $\alpha_i^r = I-1-a_i$. By Claim~\ref{claim:covering}, we 
have $\beta_{\pi_B(i)}^l = \alpha_{i}^r +1 = I-a_i$. Robot $\cB_{\pi_B(i)}$ starts at distance $a_i+2b_i$ from node $\beta_{\pi_B(i)}^l$. Since $\Delta > 3a+6b\geq 3(a_i+2b_i)$, the optimal way for robot $\cB_i$ to search its interval is to first go to the left and 
then to the right. So, we have $\beta_{\pi_B(i)}^r = I-2a_i-2b_{\pi_B(i)}-1$. By the last Claim, we have $\gamma_{\pi_C(i)}^l = \beta_{i}^r +1 = I-2a_i-2b_{\pi_B(i)}$. Robot $\cC_{\pi_C(i)}$ starts at distance $2a_i+2b_{\pi_B(i)}+4c_{\pi_C(i)}$ from node $\gamma_{\pi_C(i)}^l$. Since $\Delta > 6a+6b+12c\geq 3(2a_i+2b_{\pi_B(i)}+4c_{\pi_C(i)})$, the optimal way for robot $\cC_i$ to search its interval is to first go to the left and then to the right. Observe that since robot $\cC_{\pi_C(i)}$ must visit node $\ell$, we have : 
\begin{align*}
3I-4a_i-4b_{\pi_B(i)}-4c_{\pi_C(i)}-1  = 3I - 4S -1
\iff a_i+b_{\pi_B(i)}+c_{\pi_C(i)} = S
\end{align*}

Hence, $\pi_B,\pi_C,\pi_D$ is a solution for the instance of the $\fpartb$ problem.
\qed
\end{proof}

\begin{lemma}\label{lem:NP-DEL}
The $\DEL$ decision problem is in NP.  
\end{lemma}

\begin{proof}
We consider the verifier-based definition of NP. A certificate for the instance of the $\DEL$ decision problem is simply the set of the trajectories of the $k$ robots. Each trajectory is of length $O(n^2)$ and hence this certificate is in $O(kn^2)$ and so polynomial in the size of the instance. We can check in polynomial time (by simulating the trajectories of the robots) that every node of the line is visited before time $\Delta$ by at least $f+1$ robots. Thus, the certificate can be verified in polynomial time.
\qed
\end{proof}

\section{Illustration of the Snapshot Graph for a Ring}

\begin{figure}[!htb]
\begin{center}
\includegraphics[width=0.95\textwidth]{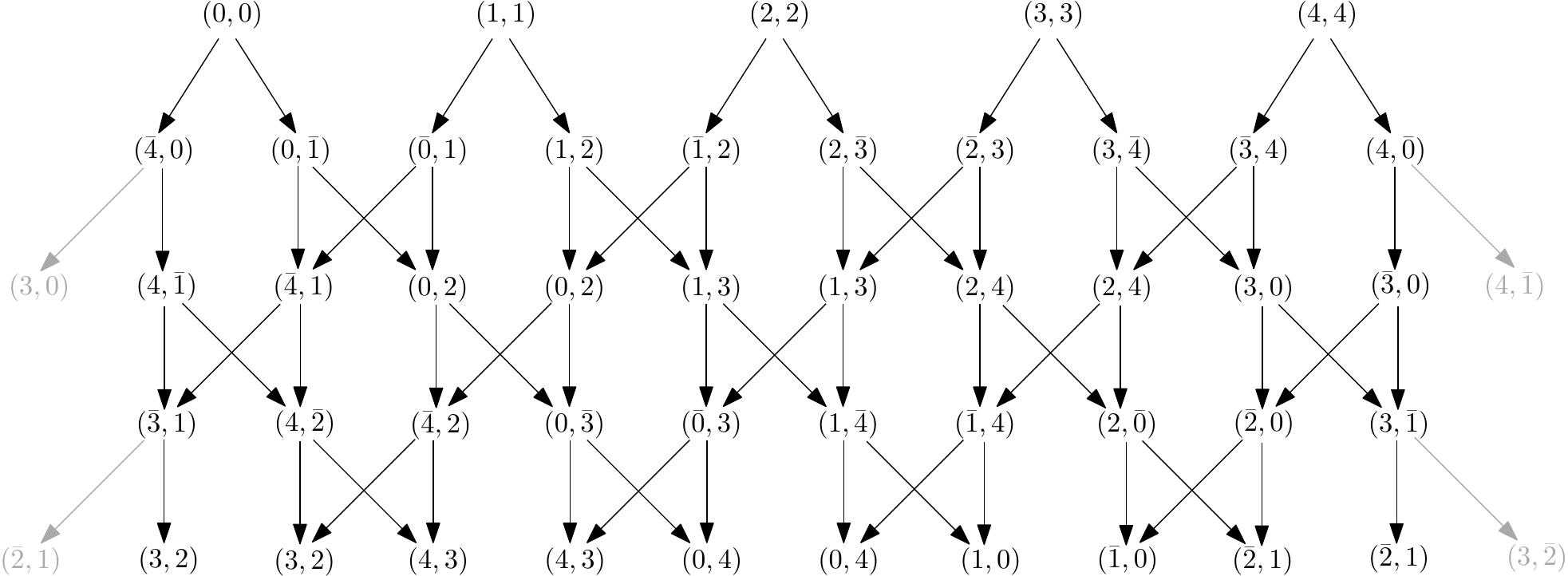}
\caption{Snapshot graph for a case of ring $R$ of five nodes. Grey nodes and edges are duplicates of other nodes at the same level (for presentation clarity). All last level nodes correspond to the ring entirely explored.}
\label{fig:snapshot-graph-ring}
\end{center}
\end{figure}

\section{Proof of Proposition~\ref{prop:RingFixedTij}}
\begin{proof}
Take a pair $i, i+1$ of successive robots around the ring $R$ for which the distance of their initial positions is the smallest. In an optimal exploration on the segment $[p_i, p_{i+1}]$ of $R$, one of its edges is idle. Knowing, which such edge is idle, we might remove it from $R$ converting the ring to a line segment. Then the line exploration Algorithm~\ref{algo:kRobotsLineFixed} may be executed for such a segment. As the segment $[p_i, p_{i+1}]$ is of size $O(n/k)$, one possible approach is to try all the possibilities of making idle every edge of $[p_i, p_{i+1}]$, each time running Algorithm~\ref{algo:kRobotsLineFixed}~for the ring segment thus obtained. This would result in overall complexity $O(n^3/k)$.

Consider the following, more careful adaptation of Algorithm~\ref{algo:kRobotsLineFixed} for the ring.  Its first part (lines 1-5) may be run once, computing all values $T_{i,j}$ in $O(n^2)$ time. Then the second part (lines 6-7) are repeated $O(n/k)$ times, i.e. for all segments $O(n/k)$ obtained from $R$ by removal of each possible idle edge between $p_i$ and $p_{i+1}$. Moreover, the min computation from line 7, by Observation~\ref{obs:logn},~may be computed in $(\log n)$ time. This results in an $O(\frac{n^2}{k} \log n)$ complexity of lines 6-7 hence in $O\left(n^2+\frac{n^2}{k} \log n \right)$ ring exploration algorithm.
\qed
\end{proof}

\section{Proof of Proposition~\ref{prop:RingFixedTij1}}

\begin{proof}

Create ring $R^{(f+1)}$ formed of $f+1$ copies of $R$, thus obtaining $k(f+1)$ possible starting positions for $k$ robots. We need to find an exploration of ring $R^{(f+1)}$ in time $T$ using $k$ robots, which may be placed at $k(f+1)$ starting positions. If such explorations are possible, then there exists one, for which each robot covers a disjoint segment of $R^{(f+1)}$, with idle edges separating them. Consider one such edge and remove it from $R^{(f+1)}$, obtaining a segment $S$ of size $n(f+1)-1$. The set of $k$ robots explore $S$ in time $T$. As the chosen idle edge belongs to some copy of ring $R$, it is sufficient to consider $n$ segments $S_0, S_1, \dots, S_{n-1}$ of size $n(f+1)-1$ and check whether one of them may be explored in time $T$.

From the corresponding snapshot graph, we compute first for any position $i$ on the ring $R^{(f+1)}$, the value $P(i)$ denoting the largest position $j$, in the counterclockwise direction around  $R^{(f+1)}$, such that a robot placed at a permitted initial position can explore in time $T$ the segment $[i,j]$ of ring $R^{(f+1)}$. Consider now an algorithm deciding for any given segment $S_m$, where $m=0,1 \dots, n-1$, whether $S_m$ may be explored in time $T$ by some set of $k$ robots, each of which may be placed at any of the given $k(f+1)$ starting positions. Starting from the initial endpoint of $S_m$, for all consecutive values of $r=1,2, \dots, k$, we compute the largest index $i_{S_m}^{r}$, such that the initial sub-segment of $S_m$ ending at node $i_{S_m}^{r}$ may be explored by a set of $r$ robots in time $T$. We can prove by induction on $r$ that 
$$
i_{S_m}^{r+1}=P(i_{S_m}^{r}+1)
$$
If $i_{S_m}^{k}$ reaches (or exceeds) the last node of segment $S_m$, then $S_m$ is explorable by $k$ robots in time $T$. 

We repeat the procedure for all segments $S_m$. As ring $R^{(f+1)}$ is possible to be explored at time $T$ if and only if one of the segments $S_m$ may be explored in time $T$ this concludes the proof.



\qed
\end{proof}

\section{Proof of Proposition~\ref{prop:star}}

\begin{proof}

We accomplish the reduction from the Partition Problem~\cite{garey2002computers}.

\probleme{Partition}
{A sets of $q$ of positive integers $A = \{a_1, a_2, \dots, a_q\}$}
{Does there exist a partition of set $A$ into two subsets of equal sum. 
}

We construct a polynomial-time reduction from the Partition problem. Consider an instance of the partition problem with the set  $A = \{a_1, a_2, \dots, a_q\}$. Let $\sum_{i=1}^q a_i = 2\sigma$. We design the corresponding instance of the star exploration problem. Consider a star consisting of $q+4$ edges $e_1, e_2, \dots, e_{q+4}$. Let the weight $w$ of each edge be such that $w(e_i)=a_i$, for $i = 1, 2, \dots, q$, and 
$w(e_{q+1})=w(e_{q+2})=w(e_{q+3})=w(e_{q+4})=4\sigma$. Take two mobile robots $1$ and $2$ and put them at the starting positions at the endpoints of edges $e_{q+1}$ and $e_{q+2}$, different from the centre of the star. Let the deadline of each node of the star be $\Delta(e_i)=10\sigma$, for  $i = 1, 2, \dots, q+4$. Note that the sum of the weights of all edges of the star equals $18\sigma$. Further, observe that each robot has to end its route at one of the edges $e_{q+3}$ and $e_{q+4}$. Indeed, otherwise one of the edges $e_{q+3}$ or $e_{q+4}$ would be traversed twice (by the same robot in both directions) and the sum of the trajectories of both robots would exceed $22\sigma$. Hence one of the robots would arrive to its last node after time $11\sigma$ and its deadline would not be met.

Consequently, robots must traverse once both edges $e_{q+1}$ and $e_{q+2}$ at the beginning of their respective routes and finish the routes by traversing edges  $e_{q+3}$ and $e_{q+4}$. Each of the remaining edges $e_i$, for $i = 1, 2, \dots, q$, must be traversed in both directions and the sum of the robot route lengths is at least $4 \cdot 4\sigma + \sum_{i=1}^q w(e_i) = 20\sigma$. In order for both robots to reach their last nodes  within their deadline time of $10\sigma$, each of them must traverse the subset of edges of total length $\sigma$. This requires solving the given instance of the partition problem. 

It is easy to see that the above reduction works not only for the star exploration from given starting positions, but also from arbitrary ones.
\qed
\end{proof}

\end{document}